\documentclass[12pt]{article}

\usepackage{amstext,amsmath,amssymb,amsfonts,bbm}
\usepackage[latin1]{inputenc}
\usepackage{epsfig}
\usepackage{dsfont}
\usepackage{hyperref}
\usepackage{amsthm}
\usepackage{subfigure}
\usepackage{color}
\usepackage{multirow}
\usepackage{psfrag}
\usepackage{graphicx}
\usepackage{extarrows}
\usepackage{braket}

\usepackage[lmargin=40pt,rmargin=40pt,tmargin=80pt,bmargin=80pt]{geometry}

\theoremstyle{plain}

\newtheorem{theorem}{Theorem}

\newtheorem{definition}{Definition}

\theoremstyle{definition}

\DeclareMathOperator{\Tr}{Tr}

\allowdisplaybreaks[4]

\newcommand{\cD}{{\cal D}}

\newcommand{\cE}{{\cal E}}
\newcommand{\cG}{{\cal G}}

\newcommand{\cJ}{{\cal J}}

\begin{document}

\title{Quenched equals annealed at leading order in the colored SYK model}

\author{\ Razvan Gurau\footnote{rgurau@cpht.polytechnique.fr, 
 Centre de Physique Th\'eorique, \'Ecole Polytechnique, CNRS, Universit\'e Paris-Saclay, F-91128 Palaiseau, France
 and Perimeter Institute for Theoretical Physics, 31 Caroline St. N, N2L 2Y5, Waterloo, ON, Canada.}}

\maketitle

\abstract{
We consider a colored version of the SYK model, that is we distinguish the $D$ vector fermionic fields involved in the interaction by a color.
We obtain the full $1/N$ series of both the quenched and annealed free energies of the model and show that 
at leading order the two are identical. The results can be used to study systematically the $1/N$ corrections to this leading order behavior.
}

\section{Introduction}

The SYK  model \cite{Sachdev:1992fk,Kitaev,Maldacena:2016hyu,Polchinski:2016xgd,Fu:2016vas} has recently attracted 
a lot of activity in the context of the $AdS/CFT$ duality.  The model can be solved (in the large $N$ limit) and 
develops a conformal symmetry in the infrared. The growth of the out of time order correlators saturates the chaos bound \cite{Maldacena:2015waa} 
suggesting that the SYK model and its variants are examples of (near) conformal field theories in one dimension dual to 
(near) extremal black holes in $AdS_2$ space.
Starting from a vector real fermion $\chi_a$ with $N$ components, the SYK model is defined by the action:
\[
 \frac{1}{2}\int_{-\beta/2}^{\beta/2} d \tau  \sum_{ a} \chi_{a}  (\tau) \partial_{\tau} \chi_{a} (\tau)  + J
 \sum_{a^1,\dots a^D} T_{a^1\dots a^D} \int_{-\beta/2}^{\beta/2} d\tau   \; \chi_{ a^1} (\tau) \dots   \chi_{a^D}(\tau) \;,
\]
where the couplings $T$ are time independent quenched random couplings with Gaussian distribution:
\[
 d\nu(T) = \bigg( \prod_{a^1,\dots a^D}  \sqrt{ \frac{N^{D-1} }{2\pi} }d T_{a^1,\dots a^D}   \bigg) \; e^{-\frac{1}{2} N^{D-1} \sum_{a^1,\dots a^D} T_{a^1\dots a^D} T_{a^1\dots a^D}  } \;.
\]
This model has a large $N$ limit dominated by melonic graphs 
\cite{Maldacena:2016hyu,Polchinski:2016xgd,Jevicki:2016bwu}. 
Random tensors \cite{review,RTM} generalize random matrices to higher dimensions and provide a new 
universality class of statistical physics models dominated by melonic graphs at large $N$.

Due to the quenching, the SYK model is not a genuine quantum mechanical model, but a statistical average of models. 
In a recent paper \cite{Witten:2016iux} Witten observed that the quenching can be eliminated if one considers a tensor model 
version of the SYK model \cite{Witten:2016iux,Gurau:2016lzk,Klebanov:2016xxf,Peng:2016mxj,Krishnan:2016bvg}:
one obtains a melonic phase in the large $N$ limit already for the annealed tensor SYK model.

In this paper we show that the quenching in the usual (vector) SYK models was in fact not needed to begin with.
We consider a colored version of the SYK model (that is we distinguish the $D$ fields in the interaction by a color index).
Building on results on random tensors \cite{RTM}, we prove that at leading order in $1/N$ the quenched and 
annealed versions of the model are identical. 
This result is quite robust: as it will be apparent from the proof we present below, it does not depend on the 
dimension (number of colors) $D$, the statistics (fermionic or bosonic) of the vector degree of freedom or 
their dynamics (propagator), provided of course that there are no ultraviolet divergences.

The colored model we discuss here has been introduced in the random tensor models context 
some years ago \cite{bonzom2013universality}. In the SYK context, it is one of the models discussed 
at length in \cite{Gross:2016kjj}. 

This result suggests that in the colored version of the vector SYK model one can promote the time independent random couplings
to a genuine field with its own dynamic. Choosing this dynamic appropriately should allow one to build also for vector SYK models
(as it was possible for tensor SYK models) true quantum field theories with a melonic phase at large $N$.
Recently a model of this kind has been discussed in \cite{Nishinaka:2016nxg}. However, such models should be treated with care: as explained
in \cite{Witten:2016iux} in this case the thermodynamic entropy of the new field (which has $N^D$ components) largely dominates that of the 
fermions (which have $DN$ components).

A similar result for the uncolored SYK model or the random energy model has been known for a long time\footnote{We thank an anonymous referee for pointing this out.} in the 
physics literature \cite{parcollet1999non,georges2000mean,Sachdev:2015efa,Michel:2016kwn}.
One argument is based on the replica trick: the diagonal ansatz for the replicas dominate the large $N$ saddle point (\emph{i.e.} in the $N \to \infty$ limit the model does not have a glassy phase) and this implies 
that the quenched and annealed averages  coincide. A second argument consists in noting that the first perturbative correction (in $J$) to the quenched average is 
subleading in $1/N$. While both arguments are of course correct, neither of the two is mathematically rigorous.

In this paper we give a mathematically rigorous proof that the quenched and annealed free energies of the colored SYK model are equal at leading order in $1/N$. This proof does not require the use of the replica trick and it holds at all orders in $J$.
The results presented in this paper go beyond the leading order in $1/N$. We derive the all order expansions in $1/N$ of the quenched and annealed free energies
which are needed in order to study the corrections to the leading order behavior like the (subleading in $1/N$) replica symmetry breaking effects.

\section{The colored SYK model}

Let us consider $D$ complex\footnote{We use complex fields to simplify the proofs, but our result should go trough with small modifications for real fermionic fields.} 
fermionic fields $\psi^1 \dots \psi^D$ living in the fundamental representation of $U(N)$ (that is each field $\psi^c$ is a vector with $N$ components). 
We call the upper index $c\in \{1,\dots D\} $ the \emph{color} of the field $\psi^c$ and we consider the action:
\[
 S = \int d\tau \bigg[ \sum_{c=1}^D  \left( \sum_{a^c} \bar \psi^c_{a^c} \frac{d}{d\tau} \psi^c_{a^c} \right) - J \sum_{a^1,\dots a^D} T_{a^1\dots a^D} \psi^1_{a^1} \dots \psi^D_{a^D} - 
\bar J  \sum_{a^1,\dots a^D}   \bar \psi^D_{a^D} \dots \bar \psi^1_{a^1} \bar T_{a^1\dots a^D} \bigg] \;.
\]
The time independent random coupling $T$ has Gaussian distribution:
\begin{align*}
& [d\bar T dT] \equiv  \bigg( \prod_{a^1,\dots a^D}  \frac{N^{D-1} d\bar T_{a^1\dots a^D} d   T_{a^1\dots a^D} }{2\pi \imath} \bigg) \;,\crcr
& d\nu(T,\bar T) =  [d\bar T dT] \; e^{- N^{D-1} \sum_{a^1,\dots a^D} T_{a^1\dots a^D} \bar T_{a^1\dots a^D}  }\;.
\end{align*}
If $D$ is even we take $T$ to be a boson, while if $D$ is odd we take it to be a fermion.

In the quenched version of the model one computes the thermal partition function $Z^T$, the thermal free energy $\ln Z^T$ and finally the averaged thermal free energy $\Braket{  \ln Z^T }$: 
\[
 Z^T = \int[d\bar \psi d\psi] \; e^{-S} \;, \qquad  \Braket{\ln Z^T} = \int  d\nu(T,\bar T) \; \ln Z^T \;.
\]
In the annealed version one computes the partition function $\Braket{Z^T} $ and the free energy $\ln \Braket{Z^T}$:
\[
 \Braket{Z^T} =\int d\nu(T,\bar T) \; Z^T \;, \qquad    \ln \Braket{Z^T} \;.
\]

In this paper we show that both $\Braket{\ln Z^T}$ and $\ln \Braket{Z^T}$ scale like $N$ and that:
\[
 \boxed{  \lim_{N\to \infty} \frac{  \Braket{\ln Z^T}  }{N} = \lim_{N\to \infty} \frac{  \ln \Braket{Z^T}  }{N}  \; .} 
\]

Similar results hold for the correlation functions of the two models: the disorder averages of the thermal correlations functions equal the 
the full correlation functions at leading order.

\section{Proof of the main result}

We first review in subsection \ref{sec:part0} some standard properties of edge colored graphs. More details about such graphs can be found in \cite{review,RTM}.
In subsections \ref{sec:part1}, \ref{sec:part2} and \ref{sec:part3} we present the results of this paper and their proofs.

\subsection{Edge colored graphs}\label{sec:part0}

We will be dealing below with $D$-- and $(D+1)$--colored graphs.

\begin{definition}\label{def:color}
 A bipartite, edge $(D+1)$--colored graph $\cG$ is a graph such that:
 \begin{itemize}
  \item the set of vertices is the disjoint union of the set of black vertices and the set of white vertices
  \item all the edges connect a black and a white vertex and have a color $c \in \{0,1, \dots D \}$
  \item all the vertices are $(D+1)$--valent and all the edges incident at a vertex have different colors
 \end{itemize}
\end{definition}

$D$--colored graphs are defined similarly, taking the set of colors to be $\{1,\dots D\}$. 
Some examples of $3$--colored graphs are presented in Fig.~\ref{fig:tensob1}.

\begin{figure}[htb]
\begin{center}
\includegraphics[width=4cm]{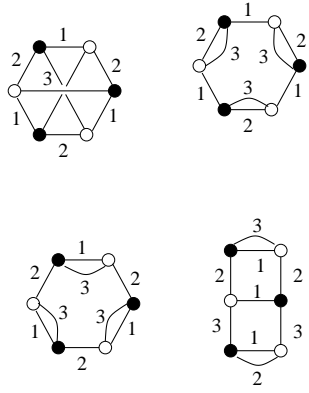}
\caption{Edge $3$-colored graphs.}
\label{fig:tensob1}
\end{center}
\end{figure}

We call the subgraphs with two colors $c_1c_2$ of $\cG$ the \emph{faces with colors $c_1c_2$} of $\cG$. 
We denote $F^{c_1c_2}(\cG)$ the number of faces with colors $c_1c_2$ of $\cG$ and $F(\cG)$ the total number of faces of $\cG$.
We furthermore denote $k(\cG)$ the number of white vertices of $\cG$ (which equals the number of black vertices of $\cG$).

\begin{theorem}\label{thm:deg}\cite{RTM,review} Let $D\ge 2$.
 For any connected $(D+1)$--colored graph $\cG$ there exists a non negative integer $\omega(\cG)$ such that:
 \begin{equation}
  F(\cG) = \frac{D(D-1)}{2} k(\cG) + D - \frac{2}{(D-1)!} \omega(\cG)  \;.
 \end{equation}
\end{theorem}

\begin{proof}
 Any cyclic permutation $\pi$ over the colors $\{0,\dots D\}$ yields an embedding of $\cG$ into a surface by arranging the edges clockwise in the order $0,\pi(0), \pi^2(0), \dots, \pi^{D}(0)$
 around the white vertices and in the reverse order around the black vertices. Each embedding is a combinatorial map $\cJ_{\pi}$ (called a jacket) whose faces are the faces with colors $\pi^q(0) \pi^{q+1}(0)$ with $q=0,\dots D$
 of $\cG$. The Euler relation for $\cJ_{\pi}$ reads:
 \[
  \sum_{q=0}^D F^{\pi^q(0)\pi^{q+1}(0)} =   (D+1) k(\cG) - 2k(\cG) + 2 - 2 g(\cJ_{\pi})\;,
 \]
where $g(\cJ_{\pi})$ is the genus of the map $\cJ_{\pi}$. Summing over the $D!$ possible cycles $\pi$ and 
taking into account that each face belongs to $2(D-1)!$ jackets (corresponding to the cycles such 
that $\pi(c_1) =c_2$ or $\pi(c_2)=c_1$ for the faces with colors $c_1c_2$)
proves the theorem with: 
\[
 \omega(\cG)  = \frac{1}{2} \sum_{\pi} g(\cJ_{\pi}) \;,
\]
which is an integer as $g(\cJ_{\pi}) = g(\cJ_{\pi^{-1}})$
 
\end{proof}

The integer $\omega(\cG)$, which is obviously an integer multiple of $(D-1)! / 2$, is called the \emph{degree} of $\cG$.
The degree of a disconnected graph is the sum of the degrees of its connected components and 
if $\cG$ has $C(\cG)$ connected components, Theorem~\ref{thm:deg} becomes:
\[
 F(\cG) = \frac{D(D-1)}{2} k(\cG) + D C(\cG)- \frac{2}{(D-1)!} \omega(\cG) \;.
\]
Observe that this relation holds also for $2$-colored graphs (which are just unions of cycles with $2$ colors), with the convention that the degree of any $2$-colored graph is $0$.

We call the connected subgraphs made by the edges with colors $\{1,\dots D\}$ of $\cG$ the $D$--\emph{bubbles} of $\cG$. They are 
$D$-colored graphs in the sense of the Definition \ref{def:color}. 

\begin{theorem}\label{thm:ineg}\cite{RTM,review}
 Let us denote $\cG\setminus \cE^0 $ the $D$-colored graph obtained from $\cG$ by deleting the edges of color $0$ (that is the disjoint union of the $D$--bubbles of $\cG$). 
  The graph $\cG$ can be connected or not. Observe that if $\cG$ is connected it can still be the case that $\cG\setminus \cE^0 $ is disconnected. We have:
 \[
  \omega(\cG) \ge D  \omega(\cG\setminus \cE^0) \;. 
 \]
\end{theorem}
\begin{proof}
Each jacket of $\cG\setminus \cE^0$ is obtained from $D$ jackets of $\cG$ by erasing the edges of color $0$ in the jacket. 
The genus of a map (which is the sum of the genera of its connected components if the map is disconnected) can only decrease by 
erasing edges. The number of faces can either increase or decrease by one when deleting an edge, and it always increases by one if the deletion separates a new connected 
component of the map. 
\end{proof}
 
It is particularly convenient to encode edge $D$--colored graphs (respectively $(D+1)$--colored graphs) in $D$-uples (respectively $(D+1)$--uples) of permutations. 
This goes as follows. Let us denote $\mathfrak{S}(k)$ the set of permutations of $k$ elements.
To any $D$--uple of permutations over $k$ elements $\sigma^{\cD} = (\sigma^{(1)},\dots \sigma^{(D)}) \in \mathfrak{S}(k)^D$
one associates the edge $D$--colored graph obtained
by drawing $k$ white vertices labeled $1,\dots k$, $k$ black vertices labeled $1,\dots k$ and connecting the white vertex $k$ to the  
black vertex $\sigma^{(c)}(k)$ by an edge of color $c$. Some examples are represented in Fig.~\ref{fig:tensobs}.

\begin{figure}[htb]
\begin{center}
\psfrag{t1}{$\sigma^{(1)} = (1)(23)$}
\psfrag{t2}{$\sigma^{(2)} = (12)(3)$}
\psfrag{t3}{$\sigma^{(3)} = (13)(2)$}
\psfrag{t11}{$\sigma^{(1)} = (1)(23)$}
\psfrag{t12}{$\sigma^{(2)} = (12)(3)$}
\psfrag{t13}{$\sigma^{(3)} = (12)(3)$}
\psfrag{t21}{$\sigma^{(1)} = (1)(23)$}
\psfrag{t22}{$\sigma^{(2)} = (12)(3)$}
\psfrag{t23}{$\sigma^{(3)} = (1)(23)$}
\psfrag{t31}{$\sigma^{(1)} = (1)(2)(3)$}
\psfrag{t32}{$\sigma^{(2)} = (12)(3)$}
\psfrag{t33}{$\sigma^{(3)} = (1)(23)$}
\includegraphics[width=8cm]{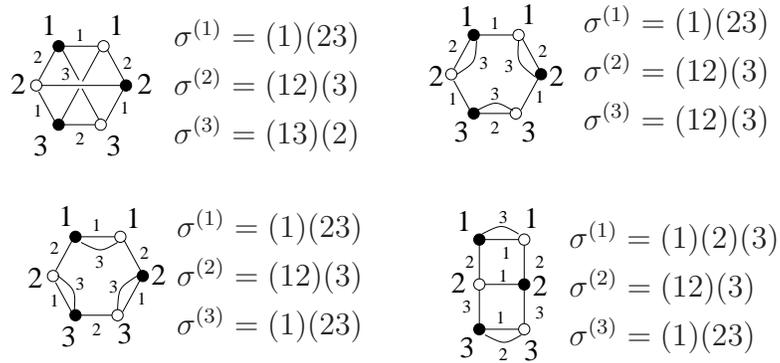}
\caption{$3$-colored graphs and associated permutations. The permutations $\sigma^{(c)}$ are written in cycle notation, the capital labels 
designate the vertices and the lower-case labels designate the colors of the edges.}
\label{fig:tensobs}
\end{center}
\end{figure}

Abusing the notation, we will denote $\sigma^{\cD}$ both the $D$--uple of permutations and the associated $D$--colored graph. 
The faces with colors $c_1c_2$ of $\sigma^{\cD}$ are the cycles of the permutation $\sigma^{(c_1)} [ \sigma^{ (c_2)} ]^{-1}$.
Edge $(D+1)$--colored graphs are 
encoded in a $(D+1)$--uple of  permutations $\sigma^{(0)} \sigma^{\cD}$.

In this notation, Theorem~\ref{thm:ineg} becomes:
\begin{equation}\label{eq:inegali}
 \omega(\sigma^{(0)}\sigma^{\cD})  \ge D  \omega(\sigma^{\cD}) \;, 
\end{equation}
and Theorem~\ref{thm:deg} applied for $(D+1)$-- and $D$--colored graphs with $k$ white vertices becomes:
\begin{align}\label{eq:faceface}
 F(\sigma^{(0)} \sigma^{\cD} ) &= \frac{D(D-1)}{2} k + D C(\sigma^{(0)}\sigma^{\cD})- \frac{2}{(D-1)!} \omega(\sigma^{(0)} \sigma^{\cD})\;, \crcr
 F(\sigma^{\cD} ) &= \frac{(D-1)(D-2)}{2} k + (D-1) C(\sigma^{\cD})- \frac{2}{(D-2)!} \omega(\sigma^{\cD})\;, 
\end{align}
where we recall that $  C(\sigma^{(0)}\sigma^{\cD})$ (respectively $C(\sigma^{\cD})$)
denotes the number of connected components of the graph $ \sigma^{(0)}\sigma^{\cD}$ (respectively $\sigma^{\cD}$).

\subsection{The thermal partition function}\label{sec:part1}

We define \cite{RTM} the \emph{trace invariant} associated to the edge colored graph $\sigma^{\cD}$ as:
\begin{align*}
 \Tr_{ \sigma^{\cD} } (T,\bar T) = \sum_{a_i^c,b_i^c}\bigg( \prod_{i=1}^k T_{a^1_i \dots a^D_i} \bar T_{b^1_i \dots b^D_i}    \bigg)
 \bigg( \prod_{c=1}^D \prod_{i=1}^c \delta_{a^c_i b^c_{\sigma^{(c)}(i)}}  \bigg) \;.
\end{align*}

In words, the trace invariant is built by putting tensors at the vertices of $ \sigma^{\cD} $ and contracting them along the edges. 
We take a tensor $T$ for every white vertex ($T_{a^1_i\dots a^D_i}$ for the white vertex $i$) and a complex conjugated tensor $\bar T$
for every black vertex ($ \bar T_{b_i^1,\dots b_i^D}$ for the black vertex $i$). Every edge in $ \sigma^{\cD}$ connects a white and a black vertex 
and has a color $c$. Say $\sigma^{(c)}(i) =j$, that is the white vertex $i$ and the black vertex $j$ are connected by an edge of color $c$. 
Corresponding to this edge, we contract the index in the position $c$ of the tensor $T_{a^1_i\dots a^D_i}$ (that is $a^c_i$) 
to the index in the position $c$ of the tensor $ \bar T_{b_j^1,\dots b_j^D}$ (that is $b^c_j$).

We denote the propagator of the SYK model by:
\[ 
C^{-1}(\tau,\tau') = \delta(\tau-\tau') \partial_{\tau'} \;,  \qquad C(\tau,\tau') = \frac{1}{2} \rm{sgn}(\tau-\tau') \;,
\]
and we define the \emph{partial amplitude} of the graph $\sigma^{\cD}$ as:
\[
A_{ \sigma^{\cD}} =  (J\bar J)^{k}\bigg( \prod_{c=1}^D \epsilon(\sigma^{(c)}) \bigg)
\int  \left( \prod_{i=1}^k   d\tau_{i} d\tau'_i \right) \;  \prod_{c=1}^D  \prod_{i=1}^k C(\tau_i,\tau'_{\sigma^{(c)} (i)}  ) \;,
\]
where $\epsilon(\sigma^{(c)}) $ is the signature of the permutation $\sigma^{(c)}$. As the moments of the Grassmann Gaussian integral for one field $\psi$ are:
\[
 \int [d\bar \psi d\psi] \; e^{ - \int d\tau d\tau' 
  \sum_a \bar \psi_a(\tau) (C^{-1})(\tau,\tau') \psi_a(\tau')} \; \prod_{i=1}^k \psi_{a_i} (\tau_i) \bar \psi_{b_i}( \tau'_i) = 
  \sum_{\sigma \in \mathfrak{S}(k)} \epsilon(\sigma) \prod_{i=1}^k  \delta_{a_i b_{\sigma(i) } } 
 C(\tau_i, \tau'_{\sigma(i)}) \;,
\]
integrating out the fermions, the thermal partition function and the thermal free energy of the colored SYK model become:
\begin{align}\label{eq:thermal}
 Z^T & = \sum_{k\ge 0 } \frac{1}{k!k!}   \sum_{\sigma^{\cD} \in \mathfrak{S}(k)^D} 
      A_{ \sigma^{\cD} } \Tr_{   \sigma^{\cD}   } (T,\bar T)  \; ,\crcr 
 \ln Z^T & = \sum_{k\ge 1 } \frac{1}{k!k!}  \sum_{ \sigma^{\cD} \in \mathfrak{S}(k)^D}^{  \sigma^{\cD} \text{ connected} } 
      A_{ \sigma^{\cD} } \Tr_{  \sigma^{\cD}  } (T,\bar T)  \;.
\end{align}

\subsection{Integrating the disorder}\label{sec:part2}

It remains to integrate out the disorder. For both the quenched and the annealed versions we need to compute:
\[
 \Braket{  \Tr_{  \sigma^{\cD}   } (T,\bar T) } = \int [d\bar T dT] \; e^{-N^{D-1} \sum_{a^1,\dots a^D} \bar T_{a^1 \dots a^D} T_{a^1\dots a^D} }   \Tr_{   \sigma^{\cD}   } (T,\bar T) \;,
\]
\emph{i.e.} the expectation of a trace invariant in a Gaussian random tensor model \cite{RTM}.
The only difference between the quenched and the annealed versions is that $ \sigma^{\cD} $ is a connected graph for the 
quenched model, while it can be disconnected in the annealed version. 

For any permutation $\sigma^{(0)}$, $[ \epsilon(\sigma^{(0)})]^D $ is $1$ for $D$ even (when $T$ is bosonic) and $\epsilon(\sigma^{(0)})$ for $D$ odd (when $T$ is fermionic) therefore the moments of the Gaussian 
tensor measure can be written for all $D$ as:
\begin{align*}
& \int [d\bar T dT] \; e^{-N^{D-1} \sum_{a^1,\dots a^D} \bar T_{a^1 \dots a^D} T_{a^1\dots a^D} }  \bigg( \prod_{i=1}^k T_{a^1_i \dots a^D_i} \bar T_{b^1_i \dots b^D_i}    \bigg)  \crcr
& = \sum_{\sigma^{(0)} \in \mathfrak{S}(k)} [ \epsilon(\sigma^{(0)})]^D \prod_{ i=1 }^k \left(  \frac{1}{ N^{D-1} } \prod_{c=1}^D \delta_{a^c_i b^c_{\sigma^{(0)}(i) } } \right) \;  .
\end{align*}

The Gaussian expectation of an invariant $\sigma^{\cD}$ will be a sum over edge $(D+1)$--colored graphs $\sigma^{(0)} \sigma^{\cD}$ with one new color, $0$.
One obtains a free sum per face with colors $0c$ for all $c = \{1,\dots D$) in the graph $ \sigma^{(0)} \sigma^{\cD}$, that is:
\begin{align*}
  \Braket{  \Tr_{  \sigma^{\cD} } (T,\bar T) } = \sum_{\sigma^{(0)} \in \mathfrak{S}(k)} [ \epsilon(\sigma^{(0)})]^D N^{-(D-1) k  + \sum_{c=1}^D F^{0c}(   \sigma^{(0)} \sigma^{\cD}  ) }  \;.
\end{align*}
We now observe that the faces with colors $0c$ of $ \sigma^{(0)} \sigma^{\cD}  $ are all the faces of $ \sigma^{(0)} \sigma^{\cD}  $ minus the faces not containing the color $0$, that is the faces of $\sigma^{\cD}$:
\[
 \sum_{c=1}^D F^{0c}(  \sigma^{(0)} \sigma^{\cD} )   = F(   \sigma^{(0)} \sigma^{\cD} )  - F( \sigma^{\cD})\;.
\]
Using Eq.~\ref{eq:faceface} for the $(D+1)$--colored graph $ \sigma^{(0)} \sigma^{\cD}$ and its $D$--colored subgraph $\sigma^{\cD}$ yields:
\begin{align*}
  & F(   \sigma^{(0)} \sigma^{\cD} )  - F( \sigma^{\cD}) \crcr
  & = (D-1)k + D C( \sigma^{(0)} \sigma^{\cD}) - (D-1) C(\sigma^{\cD}) - \frac{2}{(D-1)! } \omega(\sigma^{(0)} \sigma^{\cD} ) + \frac{2}{(D-2)!} \omega(\sigma^{\cD}) \;.
\end{align*}
This relation holds for all $D\ge 2$ (with the convention that for $D=2$, $\omega(\sigma^{\cD})=0$). We then obtain the Gaussian expectation of a trace invariant:
\begin{align*}
 \Braket{  \Tr_{  \sigma^{\cD} } (T,\bar T) } = & \sum_{\sigma^{(0)} \in \mathfrak{S}(k)} [ \epsilon(\sigma^{(0)})]^D \crcr
 & \times  N^{ D C( \sigma^{(0)} \sigma^{\cD}) - (D-1) C(\sigma^{\cD})  
    -\frac{2}{D!}  \omega(\sigma^{(0)} \sigma^{\cD} ) - \frac{2}{D(D-2)!} \bigg( \omega(\sigma^{(0)} \sigma^{\cD} ) -D \omega(\sigma^{\cD})\bigg) } \;.
\end{align*}
Combining this with Eq.~\ref{eq:thermal} we obtain for the quenched and annealed versions:
\begin{align*}
 \Braket{\ln Z^T} &=   \sum_{k\ge 1 } \frac{1}{k!k!}  \sum_{ \sigma^{(0)}\sigma^{\cD} \in \mathfrak{S}(k)^{D+1}}^{  \sigma^{\cD} \text{ connected} } 
      A_{ \sigma^{\cD} }     [ \epsilon(\sigma^{(0)})]^D
        N^{ 1-\frac{2}{D!}  \omega(\sigma^{(0)} \sigma^{\cD} ) - \frac{2}{D(D-2)!} \bigg( \omega(\sigma^{(0)} \sigma^{\cD} ) -D \omega(\sigma^{\cD})\bigg)  }  \;, \crcr
\ln\Braket{Z^T}  & = \sum_{k\ge 1 } \frac{1}{k!k!}  \sum_{ \sigma^{(0)}\sigma^{\cD}  \in \mathfrak{S}(k)^{D+1}}^{  \sigma^{(0)}\sigma^{\cD} \text{ connected} } 
      A_{ \sigma^{\cD} }     [ \epsilon(\sigma^{(0)})]^D \crcr
      & \qquad \qquad \times N^{ 1 - (D-1) \big[ C(\sigma^{\cD}) -1 \big]  -\frac{2}{D!}  \omega(\sigma^{(0)} \sigma^{\cD} ) - \frac{2}{D(D-2)!} \bigg( \omega(\sigma^{(0)} \sigma^{\cD} ) -D \omega(\sigma^{\cD})\bigg) } \; . 
\end{align*}

As the degree is non-negative and recalling Eq.~\ref{eq:inegali} we observe that these two formulae are $1/N$ expansions: the scaling with $N$ of a term is $N$ or lower.
The only difference between the two formulae is that in the quenched model, \emph{i.e.} for $ \Braket{\ln Z^T} $, $\sigma^{\cD}$ is required to be connected (which of course implies that $\sigma^{(0)}\sigma^{\cD}$ is connected), 
while in the annealed model, \emph{i.e.} for $ \ln\Braket{Z^T} $, $\sigma^{(0)}\sigma^{\cD}$ is required to be connected (while $\sigma^{\cD}$ can be disconnected).
However, if  $\sigma^{\cD}$  is disconnected, then its contribution is suppressed by a factor:
\[N^{ -(D-1)\big[ C(\sigma^{\cD}) -1 \big]   } \le N^{-(D-1)} \; .\]
It follows that:
\[
\boxed{ \lim_{N\to \infty} \frac{\Braket{\ln Z^T}}{N} = \lim_{N\to \infty} \frac{\ln\Braket{Z^T}}{N} = \sum_{k\ge 1 } \frac{1}{k!k!}  
 \sum_{ \sigma^{(0)}\sigma^{\cD} \in \mathfrak{S}(k)}^{ \genfrac{}{}{0pt}{}{  \sigma^{\cD} \text{ connected}}{  \omega(\sigma^{(0)}\sigma^{\cD} ) =0 } } 
      A_{ \sigma^{\cD} }     [ \epsilon(\sigma^{(0)})]^D    \;. } 
\]

This equality holds for all $D\ge 2$. 

\subsection{The leading order graphs}\label{sec:part3}

For all $D\ge 2$ the leading order graphs $\sigma^{(0)} \sigma^{\cD}$ are such that:
\begin{itemize}
 \item $\omega( \sigma^{(0)} \sigma^{\cD}) =0 $, which from Eq.~\ref{eq:inegali} implies that also $\omega(\sigma^{\cD})=0$
 \item $\sigma^{\cD}$ is connected, which implies that also $\sigma^{(0)} \sigma^{\cD}$ is connected
\end{itemize}

The leading order graphs must be analyzed separately for $D=2$ and $D\ge 3$.

\paragraph{The $D=2$ case.} The connected graphs with two colors are bi colored cycles. Adding the edges of color $0$ one obtains $3$-colored graphs which are just combinatorial maps. In this 
case the degree is the genus of the map. It ensues that only planar trivalent graphs with a unique face with colors $1$ and $2$ contribute in the large $N$ limit. The number of such unlabeled 
graphs with $k$ white vertices and one edge of color $1$ marked is:
\[
 \frac{1}{2k+1}\binom{2k+1}{k}
\]
Some examples are presented in Fig.~\ref{fig:leading2}

\begin{figure}[htb]
\begin{center}
\includegraphics[width=8cm]{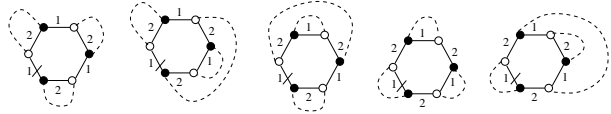}
\caption{The $5$ leading order $(2+1)$-colored graphs with three white vertices. The edges of color $0$ are represented as dashed.}
\label{fig:leading2}
\end{center}
\end{figure}

\paragraph{The $D\ge 3$  case.}

For $D\ge 3$, $(D+1)$--colored graphs of degree zero are \emph{melonic graphs} \cite{critical}. 
The first example of a melonic graph is the graph with two vertices connected by $D+1$ edges. All the other 
melonic graphs are obtained by inserting iteratively two vertices connected by $D$ edges on any of 
the available edges of a melonic graph. Observe that if $\sigma^{(0)}\sigma^{\cD}$ is melonic, 
then $\sigma^{\cD}$ is also melonic. The number of melonic $D$-colored graphs with $k$ white vertices and a marked edge of color $1$ is:
\[
 \frac{1}{Dk+1} \binom{Dk+1}{k} \;.
\]
If $\sigma^{\cD}$ is a melonic graph then there exists a unique permutation $\sigma^{(0)}$ such that $\sigma^{(0)}\sigma^{\cD}$ is melonic.
Some examples are presented in Fig.~\ref{fig:leadingD}.

\begin{figure}[htb]
\begin{center}
\includegraphics[width=8cm]{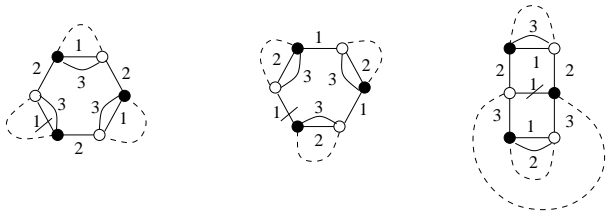}
\caption{Some examples of leading order $(3+1)$-colored graphs with three white vertices.}
\label{fig:leadingD}
\end{center}
\end{figure}

This family of graphs has been extensively studied in the literature \cite{Maldacena:2016hyu}.

\bibliography{/home/razvan/Desktop/lucru/Ongoing/Refs/Refs.bib}

\end{document}